\documentclass[letterpaper, 10pt, conference]{ieeeconf}

\IEEEoverridecommandlockouts
\usepackage[utf8]{inputenc}
\usepackage{amsmath,amssymb,amsfonts,mathtools}

\usepackage{amsthm}
\usepackage{xcolor}
\usepackage{graphicx}
\usepackage{float}
\usepackage{algorithm}
\usepackage{algorithmic}
\usepackage{booktabs}
\usepackage{caption}
\usepackage[hidelinks]{hyperref}
\usepackage{cite}

\newtheorem{assumption}{Assumption}
\newtheorem{definition}{Definition}

\newtheorem{theorem}{Theorem}
\newtheorem{corollary}{Corollary}
\newtheorem{remark}{Remark}

\newcommand{\cB}{\mathcal{B}}
\newcommand{\cX}{\mathcal{X}}
\newcommand{\cU}{\mathcal{U}}
\newcommand{\cY}{\mathcal{Y}}
\newcommand{\dW}{d_{W_1}}
\newcommand{\Vstar}{V^*}

\newcommand{\LV}{L_V}
\newcommand{\VoNI}{\mathrm{VoNI}}
\newcommand{\E}{\mathbb{E}}
\newcommand{\thetaest}{\hat{\theta}}
\newcommand{\thetatrue}{\theta_{\mathrm{true}}}

\title{Active Noise Floor Estimation for Reliability-Optimal POMDPs:\\
A Value-of-Noise-Information Approach}

\author{
Hyung-Jin Yoon
\thanks{H.-J.~Yoon is with the
Department of Mechanical and Nuclear Engineering, Tennessee Technological
University, Cookeville, TN 37135, USA.
\texttt{hjyoon@tntech.edu}}
}

\begin{document}
\maketitle

\begin{abstract}
Finite Reliability Representations (FRR) certify when a cell-constant policy is sufficient for reliable decision-making in a partially observed system with a known physical noise floor. In practice, however, the effective noise floor is often latent and context-dependent: sensing uncertainty may vary with illumination, reflectivity, or occlusion, while execution uncertainty may vary with terrain, actuator degradation, or wheel--ground interaction. This paper develops a certificate-aware active disambiguation framework for an unknown physical noise parameter $\theta=(\sigma_y,\sigma_u)$, with the sensor-only case recovered by fixing $\sigma_u$. Rather than treating noise estimation as an end in itself, the proposed framework decides when reducing posterior uncertainty in $\theta$ is valuable for preserving the FRR policy certificate. We define the \emph{Value of Noise Information} (VoNI) as the expected excess FRR certificate gap incurred when the reliability cover is calibrated to the current estimate rather than to the realized physical noise parameter. We bound VoNI using two certificate-relevant mismatch terms: action-value model mismatch and FRR radius inflation. The resulting bound shows that actively refining the noise estimate has low decision value in sub-crossover regimes where the FRR certificate is insensitive to $\theta$, but becomes valuable when posterior uncertainty can invalidate the current reliability cover. A bi-level decision maker takes a posterior over $\theta$ obtained from innovation statistics, execution residuals, or a separate online noise estimator, and triggers diagnostic probing only when uncertainty in $\theta$ threatens the FRR certificate. We also interpret VoNI as a tractable, certificate-aware approximation to a high-level finite POMDP for disambiguating latent sensing--execution regimes. Under well-specified, stationary, identifiable, and persistently exciting noise regimes, we establish posterior consistency of the noise estimate and convergence of the induced policy loss to the FRR approximation floor. Numerical results on a closed-loop unicycle-model UGV with EKF-based innovation residuals instantiate the sensor-only case and show that VoNI detects abrupt sensing-noise jumps earlier than posterior-entropy probing, tracks gradual sensing-noise drift with lower mean absolute error, and uses substantially fewer diagnostic probing actions across 50 Monte Carlo trials.
\end{abstract}

\section{Introduction}
\label{sec:intro}

Partially observed autonomous systems must plan under two distinct layers of
uncertainty: uncertainty about the current state, captured by the belief
$b \in \cB(\cX)$, and uncertainty about the physical channels through which
state information is acquired and commands are executed.
The theory of Finite Reliability Representations (FRR) \cite{frr2026}
addresses the first layer for a known physical noise floor: the deployed
sensing and execution model induces a finite cover of the reachable belief
space by reliability cells, and a cell-constant policy achieves
suboptimality bounded by $2\varepsilon/(1-\gamma)$ when every cell has
optimal action-value variation at most $\varepsilon$.
This result connects hardware sensing and actuation capability to the
minimum policy resolution needed for reliability-optimal control.

FRR, however, treats the physical noise-floor parameter as known.
This assumption is violated in practice.
The effective noise floor is a function of context: illumination, surface
reflectivity, occlusion, relative motion, terrain interaction, wheel slip, and
actuator degradation can all change the sensing and execution model in ways
that cannot be fully characterized at design time.
A stereo camera may provide accurate depth estimates on textured open surfaces
but become unreliable near reflective, transparent, or low-texture regions.
Similarly, the action actually executed by a mobile platform may differ from
the commanded action when traction, saturation, or actuator bandwidth changes.
A planner that continues to use a reliability cover calibrated at the nominal
low-noise condition may therefore operate with a cover that is under-resolved
relative to the degraded operating condition, so the original FRR certificate
no longer applies.

This paper addresses the second layer of uncertainty: uncertainty in the
physical noise floor itself. We represent the latent physical noise parameter
as
\[
  \theta := (\sigma_y,\sigma_u),
\]
where $\sigma_y$ denotes sensing noise and $\sigma_u$ denotes
action-execution or actuator noise. The sensor-only case is recovered by
fixing $\sigma_u$ and estimating only $\sigma_y$.

The question considered here is not simply how to estimate $\theta$.
Noise and covariance estimation can be performed using classical innovation
or residual-based methods \cite{mehra1970}. In recent companion work, an
online disturbance-covariance estimator was developed for Model Predictive
Path Integral (MPPI) control \cite{adaptive_mppi_noise2026}.
That companion paper treats the unknown process covariance as a statistical
estimand: it is updated cell-wise from closed-loop residuals, smoothed
spatially across neighboring cells, and propagated into an MPPI stability
certificate through an explicit adaptation penalty. Its focus is the
estimation-to-stability pipeline: given that covariance estimation is being
performed, how does the resulting estimation error affect a closed-loop MPPI
certificate?

The present paper addresses a complementary decision-theoretic question:
\emph{when does reducing uncertainty in the physical noise parameter improve
the certified reliability representation enough to justify additional
information-gathering cost?}
It does not require a particular covariance estimator or filtering method.
Instead, it assumes that a posterior or credible set over the physical noise
parameter is available. In closed-loop deployment, reducing posterior
uncertainty about $\theta$ may require diagnostic probing actions that perturb
the task trajectory, consume time or energy, or increase exposure to risk.
Thus the relevant quantity is not the absolute value of model knowledge, but
the certificate-relevant value of reducing the current posterior uncertainty.
The FRR certificate provides this criterion: information about $\theta$ is
valuable when posterior uncertainty can change the certified action-value
sensitivity or invalidate the current reliability cover. Conversely, if the
posterior credible region remains in a sub-crossover regime where the FRR
certificate is nearly insensitive to $\theta$, active probing for noise
estimation has low certificate value. In this sense, the noise estimator
supplies uncertainty over $\theta$, while VoNI decides whether that
uncertainty warrants diagnostic probing or conservative replanning.

We formalize this idea through the Value of Noise Information (VoNI), defined
as the expected excess FRR certificate gap caused by using a reliability cover
calibrated to the current estimate rather than to the realized physical noise
parameter. Importantly, VoNI is evaluated with respect to the current
posterior, not by assuming access to the realized parameter. At runtime, it is computed with respect to the current posterior
$q_t(\theta)$, which may be supplied by innovation statistics, execution
residuals, a covariance estimator, or another platform-specific noise
estimation module. VoNI therefore acts as a certificate-aware decision layer on
top of a noise estimator: the estimator supplies uncertainty over $\theta$,
while VoNI decides whether reducing that uncertainty matters for preserving
the FRR policy certificate.

VoNI drives a bi-level decision maker. The outer level maintains or receives a
posterior over $\theta$ from observation innovations and, when available,
execution residuals. The inner level executes the FRR cell-constant policy
under the current noise estimate. Diagnostic probing is triggered only when posterior uncertainty in
$\theta$ threatens the FRR certificate. The high-level decision compares the
expected reduction in the excess certificate gap with the task-level cost of
probing, rather than relying on raw posterior entropy alone.

\textbf{Contributions.}
This paper makes six contributions.
(i) We formulate active physical-noise disambiguation as a certificate-aware
decision problem rather than as generic covariance estimation.
(ii) We define the Value of Noise Information (VoNI), the expected excess FRR
certificate gap caused by using a reliability cover calibrated to an incorrect
physical noise parameter.
(iii) We extend the formulation from scalar sensing noise to a combined
sensing--execution parameter $\theta=(\sigma_y,\sigma_u)$, so that both
sensor degradation and actuator or execution uncertainty can affect the
reliability cover.
(iv) We derive an analytical VoNI bound in terms of action-value model mismatch
and FRR radius inflation, showing when posterior uncertainty in $\theta$ can
invalidate the current certificate.
(v) We interpret VoNI as a tractable, certificate-aware approximation to a
high-level finite POMDP for disambiguating latent operating regimes.
(vi) We provide a bi-level decision maker and numerical evidence on a
closed-loop unicycle-model UGV with EKF-based innovation residuals, showing
earlier detection of abrupt sensing-noise changes, improved drift tracking,
and fewer unnecessary probing actions than posterior-entropy exploration.

\textbf{Positioning.}
The proposed method sits at the intersection of POMDP approximation,
Bayesian model uncertainty, adaptive noise estimation, active perception,
hierarchical model disambiguation, and FRR. Section~\ref{sec:related_work}
positions the contribution relative to these literatures and clarifies why
VoNI is a certificate-aware value-of-information criterion rather than a
standard covariance estimator, entropy bonus, or generic POMDP solver.

\section{Related Work and Positioning}
\label{sec:related_work}

\subsection{POMDP planning and belief-space approximation}

Partially observable Markov decision processes provide the standard
formalism for sequential decision-making under hidden state uncertainty
\cite{kaelbling1998,spaan2012}. Because exact POMDP planning is generally
intractable in large or continuous domains, a large body of work develops
belief-space approximation methods, including value-function approximation,
point-based value iteration, reachability-guided belief selection, particle
representations, and belief compression
\cite{hauskrecht2000,pineau2003,kurniawati2008,roy2002}.
These methods address how to approximate planning in a known POMDP model.
The present paper addresses a different question: when the physical
noise-floor parameter that defines the POMDP model is itself uncertain, when
does disambiguating that parameter improve the reliability certificate enough
to justify information-gathering actions? VoNI is therefore not a generic
POMDP approximation method; it is an adaptive layer that decides when the
current FRR cover may be miscalibrated.

\subsection{Bayes-adaptive POMDPs and model uncertainty}

Bayes-adaptive MDP and POMDP formulations augment the state with a posterior
over unknown model parameters, thereby converting model uncertainty into a
larger belief-state planning problem \cite{ross2007bayesadaptive,doshi2010}.
Such methods provide a principled Bayesian treatment of unknown transition or
observation models, but exact planning in the augmented space is typically
expensive, and practical implementations often rely on approximations or
exploration bonuses that are not tied to a representation certificate.
In contrast, this paper does not attempt to solve the full Bayes-adaptive
POMDP over all model parameters. Instead, it focuses on the physical
noise-floor parameter
\[
  \theta=(\sigma_y,\sigma_u),
\]
and asks whether posterior uncertainty in $\theta$ can invalidate the
current FRR certificate. Thus, exploration is driven by a certificate gap,
not by posterior entropy or model uncertainty alone.

\subsection{Adaptive noise and covariance estimation}

Adaptive filtering methods estimate sensor or process-noise covariance from
innovation sequences and residual statistics, with classical work dating back
to covariance identification and adaptive Kalman filtering \cite{mehra1970}.
These methods primarily answer a statistical estimation question: how should
a noise parameter or covariance be inferred from data?

Recent companion work studied this question in the context of Model
Predictive Path Integral (MPPI) control \cite{adaptive_mppi_noise2026}. There, the unknown process-disturbance covariance is
treated as a statistical estimand, updated cell-wise from closed-loop
residuals, smoothed spatially across neighboring cells, and propagated into a
closed-loop MPPI stability certificate through an explicit adaptation penalty.
That paper therefore develops an estimation-to-stability pipeline: given that
online covariance estimation is being performed, it quantifies how estimation
error tightens or loosens an MPPI stability bound.

The present paper addresses a complementary decision problem. It does not
require a particular covariance estimator, innovation filter, or residual
model. Instead, it assumes that a posterior or credible set over the physical
noise parameter is available, and asks whether reducing that uncertainty is
worth the task-level cost. In this sense, adaptive covariance estimation
supplies uncertainty over $\theta$, while VoNI decides when that uncertainty
is large enough, in an FRR-certificate sense, to justify diagnostic probing or
conservative replanning. Thus VoNI separates the statistical question of
estimating a noise parameter from the decision-theoretic question of whether
that estimate matters for the reliability representation.

\subsection{Active perception and information-gain exploration}

Active perception, sensor management, and information-theoretic exploration
select actions that improve future observations or reduce uncertainty
\cite{lindley1956,mackay1992,krause2008,bourgault2002,charrow2015}.
These methods typically optimize entropy reduction, mutual information, or
expected information gain. Such criteria are valuable when the objective is
to learn the state, map, or model as accurately as possible. However,
information gain alone does not distinguish uncertainty that affects the
control certificate from uncertainty that is irrelevant for action selection.
VoNI uses the FRR decision-diameter criterion to make this distinction: it
values information about $\theta$ only when that information can reduce the
excess loss caused by using a reliability cover calibrated to a mismatched
physical noise regime.

Entropy-based exploration also appears in reinforcement learning, most notably
in maximum-entropy methods such as Soft Actor-Critic
\cite{haarnoja2018sac}. In those methods, entropy is applied to the action
policy to encourage stochastic exploration. The entropy baseline used in this
paper is different: it is a posterior-entropy trigger over the latent physical
noise parameter $\theta$. This baseline probes when uncertainty in $\theta$ is
large, whereas VoNI probes only when that uncertainty is expected to affect
the FRR certificate.

\subsection{Hierarchical model disambiguation}

Many robotic and machine-learning systems use hierarchical decision
architectures in which a high-level module selects modes, models, safety
policies, information-gathering behaviors, or reconfiguration actions, while
a low-level controller executes task actions \cite{sutton1999options}.
In the present paper, this structure is formalized as a finite high-level
disambiguation POMDP whose hidden state is the latent physical regime
$\theta$, whose observations are sensor-innovation and execution-residual
signatures, and whose actions include exploitation, sensing probes, actuator
probes, and conservative replanning. Solving this meta-POMDP exactly would
give the optimal disambiguation strategy, but it is generally unnecessary and
expensive. VoNI provides a tractable myopic approximation: it triggers
disambiguating actions when the posterior over $\theta$ creates a large
expected FRR certificate gap.

\subsection{Finite Reliability Representations}

The closest predecessor is the FRR framework \cite{frr2026}, which connects
belief-space resolution to physical noise floors and decision-relevant
action-value variation. FRR provides a policy-sufficiency theorem: if every
cell in a reliability cover has decision diameter at most $\varepsilon$, then
a cell-constant policy is within $2\varepsilon/(1-\gamma)$ of the optimal
policy for the same noisy POMDP. The present paper keeps this certificate but
removes the assumption that the physical noise parameter is known. The main
contribution is therefore complementary to FRR: FRR certifies a cover for a
given $\theta$, while VoNI decides when uncertainty in $\theta$ makes that
cover unreliable enough to justify active disambiguation.

\section{Problem Setup}
\label{sec:setup}

\subsection{System Model}

We adopt the FRR planning model of \cite{frr2026}, but treat the physical
noise floor as latent. In the general formulation, the physical noise
parameter is
\[
  \theta_t := (\sigma_{y,t},\sigma_{u,t}),
\]
where $\sigma_{y,t}$ denotes the sensing noise floor and $\sigma_{u,t}$
denotes the action-execution or actuator noise floor. The commanded control is
$u_t$, while the physically executed control is $\tilde u_t$, sampled from an
execution channel
\begin{equation}
  \tilde u_t \sim G_{\sigma_{u,t}}(\cdot\mid u_t).
  \label{eq:execution_channel}
\end{equation}
The system evolves as
\begin{align}
x_{t+1} &= f(x_t,\tilde u_t,w_t), \label{eq:dyn}\\
y_t     &= h(x_t)+v_t,\qquad
v_t\sim\mathcal{N}(0,\sigma_{y,t}^2 I_p).
\label{eq:obs}
\end{align}
Dynamics are Lipschitz, the reward is bounded Lipschitz
($\|r\|_\infty\!\leq r_{\max}$, with Lipschitz constant $L_r$), and the
observation map $h$ is Lipschitz. The scalar sensor-only model used in the
numerical section is recovered by fixing the execution channel
$G_{\sigma_u}$ and maintaining uncertainty only over $\sigma_{y,t}$.

The physical noise parameter may be constant, piecewise constant, or slowly
varying:
\begin{equation}
  \theta_{t+1}
  =
  \theta_t+\xi_t,\qquad
  \E[\|\xi_t\|^2]\le \delta^2,\qquad
  \theta_t\in\Theta,
\label{eq:theta_dyn}
\end{equation}
where $\Theta\subset\mathbb{R}^{d_\theta}$ is compact. In the sensor-only
case, $\theta_t=\sigma_{y,t}$. We consider two regimes: (a) \emph{step
changes}, in which $\theta_t$ is piecewise constant and changes when the
operating context changes; and (b) \emph{gradual drift}, in which
$\theta_t$ varies slowly as sensing or execution conditions change.

\begin{assumption}[Slow variation and diagnostic excitation]
\label{ass:slow}
The physical noise parameter is constant, piecewise constant with sufficiently
long dwell times, or slowly varying with $\delta\ll 1$. Moreover, on intervals
$[t_k,t_k+T]$ over which $\theta_t$ is constant or approximately constant, the
available diagnostic data are informative about all identifiable directions of
$\theta$. Specifically, let $z_\tau$ denote the noise-relevant diagnostic
signal available at time $\tau$---for example an observation innovation, an
execution residual, or both---and let $\mathcal{I}_\tau$ denote the
information used by the noise-estimation module. Define the score
\[
  s_\tau(\theta)
  :=
  \nabla_\theta
  \log p(z_\tau\mid \mathcal{I}_\tau,\theta).
\]
There exists $\delta_F>0$ such that, on each such interval,
\begin{equation}
  \lambda_{\min}
  \left(
  \sum_{\tau=t_k}^{t_k+T}
  s_\tau(\theta)s_\tau(\theta)^\top
  \right)
  \geq
  \delta_F T .
  \label{eq:diagnostic_excitation}
\end{equation}
In the sensor-only specialization, $\theta=\sigma_y$ and
\eqref{eq:diagnostic_excitation} reduces to a scalar lower bound on the
accumulated squared score of the observation-innovation likelihood.
\end{assumption}

\begin{remark}[Role of the score and the logarithm]
\label{rem:score_information}
The vector $s_\tau(\theta)$ is the score of the diagnostic likelihood. It
measures the local sensitivity of the diagnostic data distribution to changes
in the physical noise parameter. The logarithm appears because likelihoods
over time multiply, while log-likelihoods add; accumulated information over an
interval is therefore naturally represented by a sum of score outer products.
Condition~\eqref{eq:diagnostic_excitation} is a persistent-excitation
condition for the noise parameter. If the accumulated score is small in some
direction, then nearby values of $\theta$ along that direction induce nearly
indistinguishable innovation or residual statistics, and posterior
concentration cannot be expected regardless of the estimator used. The
minimum-eigenvalue form is stronger than requiring only
$\sum_\tau \|s_\tau(\theta)\|^2$ to be large: it prevents the data from being
informative in one parameter direction while leaving another direction
unidentified. This distinction is important for the combined
sensing--execution model, where the available data must distinguish changes in
$\sigma_y$ from changes in $\sigma_u$.
\end{remark}

\subsection{Joint Belief and FRR Background}
\label{subsec:frr_bg}

The ideal Bayesian information state is the joint belief
\[
  B_t(x,\theta)=p(x_t,\theta_t\mid y_{0:t},u_{0:t-1}).
\]
The proposed certificate-aware architecture does not solve the full
Bayes-adaptive POMDP over $(x,\theta)$. Instead, under
Assumption~\ref{ass:slow}, it uses the approximation
\[
  B_t(x,\theta)\approx b_t^{\hat\theta_t}(x)\,q_t(\theta),
\]
where $q_t$ is the marginal posterior or credible distribution over the
physical noise parameter, $\hat\theta_t$ is its current point estimate, and
$b_t^{\hat\theta_t}$ is the task-state belief propagated using the planning
model calibrated at $\hat\theta_t$.

\begin{assumption}[Posterior interface and calibration]
\label{ass:posterior_interface}
At each time $t$, the decision maker has access to a posterior distribution
or credible set $q_t(\theta)$ over the compact parameter set $\Theta$,
generated from diagnostic data
\[
  z_t \in \{\text{observation innovations, execution residuals, or both}\}.
\]
The diagnostic likelihood used to update $q_t$ is well specified on the
operating regime under consideration, assigns positive prior mass to a
neighborhood of the true physical noise parameter, and is locally identifiable
on intervals satisfying Assumption~\ref{ass:slow}. In the stationary case,
these conditions imply posterior concentration around the true parameter; in
the slowly varying case, they imply posterior tracking up to the drift and
excitation limits of Assumption~\ref{ass:slow}.
\end{assumption}

The posterior $q_t$ may be produced by innovation statistics, execution
residuals, an online covariance estimator, or any estimator compatible with
Assumption~\ref{ass:posterior_interface}. VoNI is evaluated with respect to
this posterior and therefore does not assume access to the true physical noise
parameter during online operation.

\begin{remark}[Scope of the posterior assumption]
VoNI is a certificate-aware decision layer, not a standalone noise-estimation
algorithm. Assumption~\ref{ass:posterior_interface} may be satisfied by
innovation-based adaptive filtering, execution-residual statistics, an online
covariance estimator, or a task-specific diagnostic model. If the diagnostic
model is misspecified, if sensing noise and execution noise are not
identifiable from the available data, or if the operating trajectory provides
insufficient excitation, then the posterior may concentrate away from the true
physical noise parameter. In that case, VoNI may underestimate the FRR
certificate risk. The posterior-calibration assumption is therefore a
substantive condition for the reliability interpretation of VoNI.
\end{remark}

For a fixed physical noise parameter $\theta=(\sigma_y,\sigma_u)$, let
$\Phi_\theta(b,\tilde u,y)$ denote the Bayesian filter update after the
executed action $\tilde u$ and observation $y$, and let
$O_\theta(\cdot\mid b,\tilde u)$ denote the predictive observation law. The
controlled belief-transition kernel induced by a commanded action $u$ is
\begin{equation}
\begin{aligned}
K_\theta(A\mid b,u)
:=
\int_{\tilde{\cU}}\int_{\cY}
&{\bf 1}_A\!\left(\Phi_\theta(b,\tilde u,y)\right)
O_\theta(dy\mid b,\tilde u) \\
&\times G_{\sigma_u}(d\tilde u\mid u),
\end{aligned}
\label{eq:belief_kernel}
\end{equation}
for measurable $A\subseteq\cB(\cX)$. This kernel is the object that appears
in the Bellman equation: it describes the distribution of the next belief
before both the executed action and the next observation are known.

A key point in \cite{frr2026} is that FRR does not require the fixed
observation update $b\mapsto \Phi_\theta(b,\tilde u,y)$ to be globally
contractive. Instead, FRR uses a reachable-set Lipschitz modulus for the
belief-transition kernel. Specifically, on the reachable belief set
$\cB_{\rm reach}$, assume that there exists a finite constant $\beta_\theta$
such that
\begin{equation}
W^{\cB}_1\!\left(
K_\theta(\cdot\mid b,u),
K_\theta(\cdot\mid b',u)
\right)
\le
\beta_\theta\,\dW(b,b')
\label{eq:kernel_lip}
\end{equation}
for all $b,b'\in\cB_{\rm reach}$ and all $u\in\cU$, where $W^{\cB}_1$ is the
Wasserstein-1 distance on probability measures over belief space, using
$\dW$ as the ground metric. If
\begin{equation}
\gamma\beta_\theta < 1,
\label{eq:discounted_kernel_condition}
\end{equation}
then the optimal action-value function is Lipschitz on $\cB_{\rm reach}$:
\begin{equation}
\begin{aligned}
\sup_{u\in\cU}
\left|
Q^*_\theta(b,u)-Q^*_\theta(b',u)
\right|
&\le
\LV(\theta)\,\dW(b,b'), \\
\LV(\theta)&:=
\frac{L_r}{1-\gamma\beta_\theta}.
\end{aligned}
\label{eq:LQ_theta}
\end{equation}
Here $\LV(\theta)$ should be read as a certified action-value sensitivity of
the noisy planning model, not as a universal monotone function of sensing
noise or execution noise alone.

\begin{assumption}[Uniform FRR regularity over plausible noise parameters]
\label{ass:uniform_frr}
For every $\theta$ in the posterior support considered by the algorithm, the
fixed-parameter POMDP induced by $K_\theta$ satisfies the FRR regularity
conditions on $\cB_{\rm reach}$. In particular, there exists
$\bar\beta<1/\gamma$ such that
\[
  \beta_\theta \le \bar\beta
  \qquad \forall \theta \in \Theta_{\rm cert},
\]
where $\Theta_{\rm cert}\subseteq\Theta$ contains the posterior credible sets
used for VoNI evaluation. Consequently,
\[
  \LV(\theta)=\frac{L_r}{1-\gamma\beta_\theta}
\]
is finite and uniformly bounded on $\Theta_{\rm cert}$. Moreover,
$\LV(\theta)$ and the fixed-parameter action-value functions $Q^*_\theta$ are
continuous in $\theta$ on $\Theta_{\rm cert}$, or admit certified upper bounds
sufficient for evaluating the VoNI bound.
\end{assumption}

An FRR cover at level $(\theta,\varepsilon)$ is a finite collection of
reliability cells $\{C_i,\hat b_i\}_{i=1}^N$ covering $\cB_{\rm reach}$ such
that
\begin{equation}
\Delta_Q^\theta(C_i)
:=
\sup_{b,b'\in C_i}
\sup_{u\in\cU}
\left|
Q^*_\theta(b,u)-Q^*_\theta(b',u)
\right|
\le \varepsilon .
\label{eq:decision_diameter_theta}
\end{equation}
A sufficient construction is to use cells satisfying
\begin{equation}
\operatorname{diam}_{\dW}(C_i)
\le
\frac{\varepsilon}{\LV(\theta)} .
\label{eq:frr_radius_theta}
\end{equation}
The corresponding cell-constant policy $\bar\pi_{\theta,\varepsilon}$ chooses
an action using the representative belief of the cell containing $b$. The FRR
policy-sufficiency theorem gives
\begin{equation}
\left\|
V^*_\theta
-
V^{\bar\pi_{\theta,\varepsilon}}_\theta
\right\|_\infty
\le
\frac{2\varepsilon}{1-\gamma}.
\label{eq:frr_policy_bound}
\end{equation}

The information crossover scale used below denotes a model-dependent region
of the noise-parameter space in which variations in $\theta$ induce
certificate-relevant changes in the FRR constants. It is not assumed to be a
universal physical constant. Rather, it is obtained from offline calibration,
simulation, or certified bounds on the fixed-parameter planning models
$\{K_\theta,Q^*_\theta,\LV(\theta)\}_{\theta\in\Theta_{\rm cert}}$. In the
sensor-only UGV calibration used in the numerical section, the crossover scale
is denoted $\sigma_0=0.3\,\mathrm{m}$. The certified sensitivities satisfy
$\LV(0.05)\approx 1.03$ under the nominal sensing condition and
$\LV(0.80)\approx 7.97$ under the degraded sensing condition. Since the
sufficient FRR cell radius scales as
\[
  r_{\rm FRR}(\theta)=\frac{\varepsilon}{\LV(\theta)},
\]
a cover calibrated at $\sigma_y=0.05\,\mathrm{m}$ can use cells roughly
$7.97/1.03\approx 7.7$ times larger than the cells certified at
$\sigma_y=0.80\,\mathrm{m}$. Thus an undetected transition from the nominal
to the degraded sensing regime can invalidate the nominal reliability cover,
even if the task-state belief update itself remains well defined. Conversely,
if the posterior support lies entirely in a sub-crossover region where
$\LV(\theta)$ is nearly constant, then reducing uncertainty in $\theta$ has
little effect on the certified FRR radius and has low VoNI.

\begin{remark}[Crossover scale and the role of the assumptions]
\label{rem:crossover_assumptions}
The crossover scale illustrates why the assumptions above are needed. 
Assumption~\ref{ass:uniform_frr} requires that all plausible fixed-parameter
models in $\Theta_{\rm cert}$ admit finite FRR certificates, so that
$\LV(\theta)$ and the corresponding reliability radii can be compared across
posterior samples. Assumption~\ref{ass:posterior_interface} requires that the
posterior over $\theta$ be calibrated enough for these comparisons to be
meaningful. Assumption~\ref{ass:slow} requires that the diagnostic data contain
sufficient information to distinguish certificate-relevant regimes before the
noise parameter changes too rapidly. These conditions are not merely technical:
if $\LV(\theta)$ is not certified over the posterior support, then VoNI cannot
bound the certificate loss; if the posterior is miscalibrated, then VoNI may
assign low value to a truly degraded regime; and if the diagnostic signal
cannot distinguish sensing noise from execution noise, then the combined
parameter $\theta=(\sigma_y,\sigma_u)$ is not identifiable. In the full
sensing--execution model, the same crossover logic applies to either component
of $\theta$ through its effect on the belief kernel $K_\theta$, the
fixed-parameter value functions $Q^*_\theta$, and the certified sensitivity
$\LV(\theta)$.
\end{remark}

\section{Value of Noise Information}
\label{sec:voni}

\subsection{Definition}

For a fixed physical noise parameter $\theta$, the FRR
policy-sufficiency theorem gives
\begin{equation}
V^*_\theta(b)
-
V^{\bar\pi_{\theta,\varepsilon}}_\theta(b)
\le
\frac{2\varepsilon}{1-\gamma},
\qquad
b\in\cB_{\rm reach}.
\label{eq:frr_pointwise_bound}
\end{equation}
Thus $2\varepsilon/(1-\gamma)$ is the certified approximation floor incurred
by using a cell-constant FRR policy when the reliability cover is calibrated
to the same physical noise parameter that defines the planning model.

The difficulty addressed here is that the planner may instead use a cover
calibrated to the current estimate $\thetaest$, while the candidate physical
noise parameter governing the model is $\theta\neq\thetaest$. In that case,
the policy $\bar\pi_{\thetaest,\varepsilon}$ is not the policy certified by
\eqref{eq:frr_pointwise_bound} for the model $\theta$, and the nominal FRR
certificate may no longer apply. We therefore measure only the loss that
exceeds the FRR approximation floor:
\begin{equation}
\mathcal G_\varepsilon(b,\thetaest;\theta)
:=
\left[
V^*_\theta(b)
-
V^{\bar{\pi}_{\thetaest,\varepsilon}}_\theta(b)
-
\frac{2\varepsilon}{1-\gamma}
\right]_+ .
\label{eq:excess_certificate_gap}
\end{equation}
Here $[z]_+:=\max\{z,0\}$ denotes the positive part. Thus
$\mathcal G_\varepsilon$ is zero whenever the policy calibrated at
$\thetaest$ remains within the nominal FRR tolerance for the model $\theta$,
and is positive only when noise-parameter mismatch creates an additional
certificate gap.

\begin{definition}[Value of Noise Information]
\label{def:voni}
Fix a decision tolerance $\varepsilon>0$ and a posterior distribution $q$ over
the physical noise parameter. For belief $b\in\cB_{\rm reach}$ and current
estimate $\thetaest$, define
\begin{equation}
\VoNI_\varepsilon(b,\thetaest;q)
:=
\E_{\theta\sim q}
\left[
V^*_\theta(b)
-
V^{\bar{\pi}_{\thetaest,\varepsilon}}_\theta(b)
-
\frac{2\varepsilon}{1-\gamma}
\right]_+ ,
\label{eq:voni}
\end{equation}
equivalently,
\[
  \VoNI_\varepsilon(b,\thetaest;q)
  =
  \E_{\theta\sim q}
  \left[
  \mathcal G_\varepsilon(b,\thetaest;\theta)
  \right].
\]
Here $\bar{\pi}_{\thetaest,\varepsilon}$ is the cell-constant FRR policy
produced from a cover calibrated at $\thetaest$ and tolerance $\varepsilon$.
\end{definition}

VoNI is therefore not the ordinary FRR approximation error. That baseline
error is already accounted for by \eqref{eq:frr_pointwise_bound}. VoNI is the
posterior expected excess loss beyond the nominal FRR certificate caused by
using a reliability cover calibrated to a possibly incorrect physical noise
parameter. Consequently, if the posterior concentrates on the current
estimate, $q=\delta_{\thetaest}$, then
$\VoNI_\varepsilon(b,\thetaest;q)=0$ by the FRR policy-sufficiency theorem.

\subsection{Analytical Bound}

To bound VoNI, we introduce two mismatch quantities. First, define the uniform
action-value model mismatch
\begin{equation}
\Delta_Q(\theta,\thetaest)
:=
\sup_{b\in\cB_{\rm reach}}
\sup_{u\in\cU}
\left|
Q^*_\theta(b,u)-Q^*_{\thetaest}(b,u)
\right|.
\label{eq:Delta_Q}
\end{equation}
This term measures the change in the fixed-parameter optimal action-value
function when the planning model is changed from $\thetaest$ to $\theta$.

Second, define the FRR radius inflation factor
\begin{equation}
\rho(\theta,\thetaest)
:=
\frac{\LV(\theta)}{\LV(\thetaest)}.
\label{eq:rho_theta}
\end{equation}
If $\rho(\theta,\thetaest)>1$, then a cover calibrated at radius
$\varepsilon/\LV(\thetaest)$ is under-resolved relative to the sufficient FRR
radius $\varepsilon/\LV(\theta)$ for the model $\theta$. If
$\rho(\theta,\thetaest)\le 1$, then the $\thetaest$-calibrated radius is at
least as fine as the sufficient radius associated with $\theta$, so there is
no radius-induced loss of resolution.

\begin{theorem}[Posterior VoNI bound]
\label{thm:voni_bound}
Fix a belief $b\in\cB_{\rm reach}$, an estimate $\thetaest$, and a posterior
$q$ over $\Theta_{\rm cert}$. Suppose the FRR Lipschitz certificate
\eqref{eq:LQ_theta} holds for $\thetaest$ and for $q$-almost every
$\theta\sim q$. Suppose also that the $\thetaest$-calibrated FRR cover is
constructed with cell diameter at most $\varepsilon/\LV(\thetaest)$.
Then
\begin{equation}
\VoNI_\varepsilon(b,\thetaest;q)
\le
\frac{2}{1-\gamma}
\E_{\theta\sim q}
\left[
\Delta_Q(\theta,\thetaest)
+
\varepsilon\bigl(\rho(\theta,\thetaest)-1\bigr)
\right]_+ .
\label{eq:voni_bound}
\end{equation}
\end{theorem}

\begin{proof}
Fix any $\theta$ in the posterior support for which the FRR Lipschitz
certificate holds. Let $C_i$ be the $\thetaest$-calibrated FRR cell containing
$b$, with representative $\hat b_i$. By construction,
\[
  \dW(b,\hat b_i)
  \le
  \frac{\varepsilon}{\LV(\thetaest)} .
\]
Evaluating the same cell under the model $\theta$ gives
\[
\sup_{u\in\cU}
\left|
Q^*_\theta(b,u)-Q^*_\theta(\hat b_i,u)
\right|
\le
\LV(\theta)\dW(b,\hat b_i)
\le
\varepsilon \rho(\theta,\thetaest).
\]
Let $u^*_\theta(b)\in\arg\max_u Q^*_\theta(b,u)$ and let
\[
\bar u
=
\bar\pi_{\thetaest,\varepsilon}(b)
\in
\arg\max_{u\in\cU} Q^*_{\thetaest}(\hat b_i,u).
\]
Using the preceding cell-variation bound and the definition of
$\Delta_Q(\theta,\thetaest)$,
\begin{align*}
Q^*_\theta(b,\bar u)
&\ge
Q^*_\theta(\hat b_i,\bar u)
-
\varepsilon\rho(\theta,\thetaest)                                      \\
&\ge
Q^*_{\thetaest}(\hat b_i,\bar u)
-
\Delta_Q(\theta,\thetaest)
-
\varepsilon\rho(\theta,\thetaest)                                      \\
&\ge
Q^*_{\thetaest}(\hat b_i,u^*_\theta(b))
-
\Delta_Q(\theta,\thetaest)
-
\varepsilon\rho(\theta,\thetaest)                                      \\
&\ge
Q^*_\theta(\hat b_i,u^*_\theta(b))
-
2\Delta_Q(\theta,\thetaest)
-
\varepsilon\rho(\theta,\thetaest)                                      \\
&\ge
Q^*_\theta(b,u^*_\theta(b))
-
2\Delta_Q(\theta,\thetaest)
-
2\varepsilon\rho(\theta,\thetaest).
\end{align*}
Therefore $\bar\pi_{\thetaest,\varepsilon}$ is
$2\Delta_Q(\theta,\thetaest)+2\varepsilon\rho(\theta,\thetaest)$-greedy with
respect to $Q^*_\theta$. The standard approximate-greedy performance bound
then gives
\begin{equation}
V^*_\theta(b)-V^{\bar\pi_{\thetaest,\varepsilon}}_\theta(b)
\le
\frac{2}{1-\gamma}
\left[
\Delta_Q(\theta,\thetaest)+\varepsilon\rho(\theta,\thetaest)
\right].
\label{eq:mismatch_policy_bound}
\end{equation}
Subtracting the nominal FRR certificate $2\varepsilon/(1-\gamma)$ and taking
positive parts yields
\[
\begin{aligned}
&\left[
V^*_\theta(b)
-
V^{\bar\pi_{\thetaest,\varepsilon}}_\theta(b)
-
\frac{2\varepsilon}{1-\gamma}
\right]_+ \\
&\le
\frac{2}{1-\gamma}
\left[
\Delta_Q(\theta,\thetaest)
+
\varepsilon\bigl(\rho(\theta,\thetaest)-1\bigr)
\right]_+ .    
\end{aligned}
\]
Taking expectation with respect to $\theta\sim q$ gives
\eqref{eq:voni_bound}.
\end{proof}

Theorem~\ref{thm:voni_bound} separates two sources of certificate mismatch.
The term $\Delta_Q(\theta,\thetaest)$ measures model mismatch at the level of
optimal action values. The factor $\rho(\theta,\thetaest)$ measures resolution
mismatch: it compares the certified sensitivity under the candidate model
$\theta$ with the sensitivity used to construct the current cover. When
$\rho(\theta,\thetaest)>1$, the cover calibrated at $\thetaest$ may fail to
meet the sufficient FRR radius for $\theta$. When
$\rho(\theta,\thetaest)\le 1$, the current cover is not under-resolved
relative to $\theta$, and any remaining excess gap must come from
action-value model mismatch rather than radius inflation.

\subsection{Exploration Threshold}

The posterior bound in Theorem~\ref{thm:voni_bound} gives a certificate-level
interpretation of the crossover scale. Exploration is valuable not merely
because the posterior over $\theta$ is diffuse, but because the posterior
places mass on parameter values that can change either the fixed-parameter
action-value function or the FRR resolution required for certification.

The next corollary is a local specialization of
Theorem~\ref{thm:voni_bound}. It shows that when the posterior credible set
lies inside a region where both the value model and the FRR sensitivity vary
slowly with $\theta$, VoNI is controlled by the posterior radius.

\begin{corollary}[Low VoNI in a sub-crossover credible regime]
\label{cor:threshold}
Let $\Theta_-\subset\Theta_{\rm cert}$ be a sub-crossover credible region,
for example a region in which both sensing and execution uncertainty remain
near their nominal values. Suppose $q(\Theta_-)=1$ and
$\thetaest\in\Theta_-$. Assume that there exist finite local sensitivity
constants $M_-$ and $\ell_-$ such that, for all $\theta\in\Theta_-$,
\begin{align}
\Delta_Q(\theta,\thetaest)
&\le
M_-\|\theta-\thetaest\|,
\label{eq:local_delta_bound}\\
|\LV(\theta)-\LV(\thetaest)|
&\le
\ell_-\|\theta-\thetaest\|.
\label{eq:local_lip_bound}
\end{align}
Then, for every $b\in\cB_{\rm reach}$,
\begin{equation}
\VoNI_\varepsilon(b,\thetaest;q)
\le
\frac{2}{1-\gamma}
\left(
M_-+\frac{\varepsilon\ell_-}{\LV(\thetaest)}
\right)
\E_{\theta\sim q}\!\left[\|\theta-\thetaest\|\right].
\label{eq:sub_crossover_bound}
\end{equation}
\end{corollary}

\begin{proof}
Since $q(\Theta_-)=1$, it is sufficient to bound the integrand in
Theorem~\ref{thm:voni_bound} for $\theta\in\Theta_-$. For such $\theta$,
\[
\rho(\theta,\thetaest)-1
=
\frac{\LV(\theta)-\LV(\thetaest)}{\LV(\thetaest)}.
\]
Hence
\[
\begin{aligned}
&\left[
\Delta_Q(\theta,\thetaest)
+
\varepsilon\bigl(\rho(\theta,\thetaest)-1\bigr)
\right]_+ \\
&\quad\le
\Delta_Q(\theta,\thetaest)
+
\varepsilon
\frac{|\LV(\theta)-\LV(\thetaest)|}{\LV(\thetaest)} .
\end{aligned}
\]
Using \eqref{eq:local_delta_bound} and \eqref{eq:local_lip_bound}, we obtain
\[
\left[
\Delta_Q(\theta,\thetaest)
+
\varepsilon\bigl(\rho(\theta,\thetaest)-1\bigr)
\right]_+
\le
\left(
M_-+\frac{\varepsilon\ell_-}{\LV(\thetaest)}
\right)
\|\theta-\thetaest\|.
\]
Taking expectation in Theorem~\ref{thm:voni_bound} gives
\eqref{eq:sub_crossover_bound}.
\end{proof}

\begin{remark}[Constants in the sub-crossover bound]
\label{rem:sub_crossover_constants}
The constants $M_-$ and $\ell_-$ are not algorithmic tuning parameters and are
not required to compute the VoNI score. They are local certification constants
used only to obtain the simplified sub-crossover estimate
\eqref{eq:sub_crossover_bound}. In implementation, the decision maker may
compute or upper-bound VoNI directly from posterior samples using
Theorem~\ref{thm:voni_bound}, through the quantities
$\Delta_Q(\theta,\thetaest)$ and $\rho(\theta,\thetaest)$. The constants
$M_-$ and $\ell_-$ summarize these quantities only in the special case where
the posterior support lies in a locally regular, sub-crossover region.
\end{remark}

Corollary~\ref{cor:threshold} is the operational threshold statement. When
the posterior credible region remains in a sub-crossover regime and the
certified sensitivity $\LV(\theta)$ is nearly flat there, reducing posterior
uncertainty in $\theta$ produces only a small certificate improvement. In
contrast, if the posterior assigns nonnegligible mass to values where sensing
or execution degradation changes $Q^*_\theta$ or $\LV(\theta)$, the two terms
in Theorem~\ref{thm:voni_bound} can become large, and diagnostic probing or
conservative replanning can be certificate-relevant.

This distinction is important for abrupt changes. Before a noise-floor jump
is reflected in the posterior, the posterior may remain concentrated in the
nominal sub-crossover regime, and VoNI suppresses unnecessary probing. Once
innovation or execution residuals make the nominal model implausible,
posterior mass moves toward degraded-noise values; the bound
\eqref{eq:voni_bound} then increases and triggers probing only when the FRR
certificate is at risk.

Figure~\ref{fig:ugv_certificate_landscape} plots the conditional
certificate-gap proxy as a function of the sensing estimate $\hat\sigma_y$
for several sensing-noise regimes in the sensor-only UGV specialization. The
curves illustrate the analytical prediction of
Theorem~\ref{thm:voni_bound} and Corollary~\ref{cor:threshold}: local
sub-crossover uncertainty has low certificate value, whereas posterior mass
near degraded-noise regimes produces a large expected excess certificate gap.

\subsection{VoNI as a Myopic High-Level POMDP Approximation}
\label{subsec:meta_pomdp}

The preceding bound can be interpreted as a tractable approximation to a
higher-level decision problem over the latent physical noise regime. This
interpretation clarifies the role of VoNI in the overall architecture: the
low-level planner executes an FRR policy in belief space, while a high-level
decision maker determines whether the current posterior over $\theta$ is
sufficiently certificate-relevant to justify probing, conservative replanning,
or continued exploitation.

Consider a finite discretization of the plausible noise-parameter set,
\[
  \Theta=\{\theta^1,\ldots,\theta^M\}\subset\Theta_{\rm cert}.
\]
The high-level hidden state is the latent physical regime $\theta_t$, not the
task state $x_t$. The high-level belief is
\[
  q_t(i)=\Pr(\theta_t=\theta^i\mid z^H_{0:t},a^H_{0:t-1}),
\]
where $z_t^H$ denotes a residual-level diagnostic observation, such as a
quantized observation innovation, an execution residual, or a joint residual
signature. The high-level action set may include
\[
  \mathcal A^H
  =
  \{\mathrm{exploit},\mathrm{sense\mbox{-}probe},
    \mathrm{actuator\mbox{-}probe},\mathrm{safe\mbox{-}replan}\}.
\]
The action $\mathrm{exploit}$ executes the low-level FRR policy calibrated to
the current estimate $\thetaest_t$. The probing actions choose behaviors that
increase information about sensing or execution uncertainty, while
$\mathrm{safe\mbox{-}replan}$ switches to a conservative reliability cover,
for example one calibrated to a high-noise credible bound.

The corresponding high-level POMDP is specified by a regime transition model
\[
  P^H(\theta'\mid\theta,a^H),
\]
a residual observation model
\[
  O^H(z^H\mid\theta,a^H),
\]
and a reward depending on the current low-level belief $b_t$. For the exploit
action, the regime-conditioned reward is
\[
  R^H(b_t,\theta,\mathrm{exploit})
  =
  V^{\bar\pi_{\thetaest_t,\varepsilon}}_\theta(b_t),
\]
because the system acts using the FRR cover calibrated at $\thetaest_t$ while
the candidate regime is $\theta$. For a probing action $a^H$, the reward
includes the value of the probing policy and its immediate cost:
\[
  R^H(b_t,\theta,a^H)
  =
  V^{\pi^{\mathrm{probe}}_{a^H}}_\theta(b_t)-c(a^H).
\]
The high-level Bellman equation can be written compactly as
\begin{equation}
J(b,q)
=
\max_{a^H\in\mathcal A^H}
Q^H(b,q,a^H),
\label{eq:meta_pomdp_bellman}
\end{equation}
where
\begin{equation}
Q^H(b,q,a^H)
:=
\bar R^H(b,q,a^H)
+
\gamma_H\bar J^H(b,q,a^H).
\label{eq:meta_q_function}
\end{equation}
The two terms are
\begin{equation}
\bar R^H(b,q,a^H)
:=
\sum_{\theta\in\Theta}
q(\theta)R^H(b,\theta,a^H),
\label{eq:meta_expected_reward}
\end{equation}
and
\begin{equation}
\bar J^H(b,q,a^H)
:=
\sum_{z^H\in\mathcal Z^H}
P(z^H\mid q,a^H)
J\!\left(
b^+,
\tau^H(q,a^H,z^H)
\right).
\label{eq:meta_expected_continuation}
\end{equation}
Here $\tau^H$ is the Bayesian update of the high-level belief and
\begin{equation}
P(z^H\mid q,a^H)
=
\sum_{\theta\in\Theta}
q(\theta)O^H(z^H\mid\theta,a^H).
\label{eq:meta_predictive_obs}
\end{equation}

Solving \eqref{eq:meta_pomdp_bellman} exactly is generally unnecessary for
the FRR objective. The quantity needed by the low-level certificate is not the
full value of model information, but the expected excess loss caused by using
a reliability cover calibrated to the current estimate. For the exploit
action, this excess loss is precisely
\[
\E_{\theta\sim q}
\left[
V^*_\theta(b)
-
V^{\bar\pi_{\thetaest,\varepsilon}}_\theta(b)
-
\frac{2\varepsilon}{1-\gamma}
\right]_+
=
\VoNI_\varepsilon(b,\thetaest;q).
\]
Thus VoNI is the certificate-aware one-step regret of exploiting the current
cover, after subtracting the nominal FRR approximation floor.

This observation yields a myopic approximation to the high-level POMDP. A
diagnostic action is certificate-relevant when its expected reduction in VoNI
exceeds its task-level cost. Define the expected posterior VoNI after action
$a^H$ as
\begin{equation}
\begin{aligned}
&\overline{\VoNI}^{+}_{\varepsilon} \\
&:=
\E_{z^H\sim P(\cdot\mid q,a^H)}
\!\left[
\VoNI_\varepsilon
\left(
b^+,
\thetaest^+(z^H),
\tau^H(q,a^H,z^H)
\right)
\right].
\end{aligned}
\label{eq:expected_future_voni}
\end{equation}
Then the myopic VoNI reduction proxy is
\begin{equation}
\mathcal I_\varepsilon^H(b,q,a^H)
:=
\VoNI_\varepsilon(b,\thetaest;q)
-
c(a^H)
-
\overline{\VoNI}^{+}_{\varepsilon}.
\label{eq:myopic_voni_reduction}
\end{equation}
A positive value of $\mathcal I_\varepsilon^H$ indicates that the expected
certificate improvement from disambiguating $\theta$ exceeds the probing cost.
The algorithm developed below uses this principle in simplified form: it
triggers probing or conservative replanning only when the posterior over
$\theta$ induces a sufficiently large expected FRR certificate gap.

This interpretation connects VoNI directly to the analytical bound in
Theorem~\ref{thm:voni_bound}. The theorem provides a computable upper bound on
the exploit-regret term in \eqref{eq:myopic_voni_reduction}, while
Corollary~\ref{cor:threshold} explains why the same term is small in
sub-crossover credible regimes. Consequently, the high-level POMDP need not be
solved exactly: VoNI supplies a certificate-aware surrogate for the part of
the high-level decision problem that matters for FRR reliability.

\section{Bi-Level Decision Maker}
\label{sec:bilevel}

\subsection{Outer Level: Physical Noise-Floor Posterior Interface}

In the general case, the posterior $q_t(\theta)$ may be updated from both
observation innovations and execution residuals. Observation innovations
provide information about $\sigma_y$, while discrepancies between commanded
and realized motion provide information about $\sigma_u$. For example, if an
estimate of the realized motion or executed action is available, residuals of
the form
\[
  e_t = \hat x_{t+1}-f(\hat x_t,u_t,0)
\]
can be used to update the execution-noise component. The precise likelihood
and posterior representation are platform-dependent; the role of
$q_t(\theta)$ is to quantify uncertainty in the physical parameter that
determines $K_\theta$, $Q_\theta^*$, and $\LV(\theta)$.

A conjugate inverse-Gamma update is one possible posterior interface for a
stationary Gaussian sensing-noise model. The numerical study in
Section~\ref{sec:experiments}, however, uses the sensor-only specialization
$\theta_t=\sigma_{y,t}$ with an innovation-based adaptive estimate and a
truncated-Gaussian approximation to its uncertainty. At each step, the
innovation is
\begin{equation}
  \nu_t = y_t-h(\hat x_t^-),
  \label{eq:innov}
\end{equation}
and a corrected innovation statistic accounts approximately for the
state-belief contribution to the innovation covariance. Posterior samples
from the resulting approximation are then propagated through the calibrated
certificate map to evaluate the VoNI bound. The general analysis below is
conditional on a posterior interface satisfying
Assumption~\ref{ass:posterior_interface}; it does not require the particular
finite-horizon estimator used in the numerical example.

\subsection{Inner Level: FRR Policy Execution}

Given the current posterior estimate $\thetaest_t$, the inner level executes
the cell-constant policy $\bar{\pi}_{\thetaest_t,\varepsilon}$ from the FRR
cover calibrated at $\thetaest_t$. The FRR cover is recomputed whenever
\[
  |\LV(\thetaest_t)-\LV(\thetaest_{t-1})|>\eta
\]
for a prescribed cover-update threshold $\eta>0$. The two-timescale structure
arises naturally: the physical noise parameter and its posterior typically
evolve more slowly than the task-state belief $b_t$, which is updated at every
step.

\subsection{VoNI-Guided High-Level Action Rule}
\label{subsec:obj}

For each diagnostic action
$a^H\in\mathcal A_{\rm probe}^H$, the high-level decision maker computes or
approximates the myopic certificate improvement
$\mathcal I_\varepsilon^H(b_t,q_t,a^H)$ defined in
\eqref{eq:myopic_voni_reduction}. The selected high-level action is
\begin{align}
a_{t,\mathrm{probe}}^H
&\in
\arg\max_{a^H\in\mathcal A_{\mathrm{probe}}^H}
\mathcal I_\varepsilon^H(b_t,q_t,a^H),
\label{eq:best_probe_action}\\
I_t^\star
&:=
\max_{a^H\in\mathcal A_{\mathrm{probe}}^H}
\mathcal I_\varepsilon^H(b_t,q_t,a^H),
\label{eq:best_probe_value}
\end{align}
and select
\begin{equation}
a_t^H
=
\begin{cases}
a_{t,\mathrm{probe}}^H, & I_t^\star>0,\\
\mathrm{exploit}, & I_t^\star\le 0.
\end{cases}
\label{eq:high_level_action_rule}
\end{equation}
Thus a diagnostic action is chosen only when its expected reduction in the
excess FRR certificate gap exceeds its immediate task-level cost. This rule
does not reward a large VoNI value; rather, it values an action that is
expected to reduce VoNI.

Algorithm~\ref{alg:voni} summarizes the resulting bi-level architecture.
When the action-conditioned expectation in
\eqref{eq:myopic_voni_reduction} is expensive, it may be replaced by a
one-step approximation based on posterior samples and the bound in
Theorem~\ref{thm:voni_bound}. The numerical example uses such a smooth
score-based approximation: the current posterior VoNI proxy is mapped to a
stochastic probing probability through a sigmoid rule.

\begin{algorithm}[t]
\caption{VoNI-Guided Bi-Level Decision Maker}
\label{alg:voni}
\begin{algorithmic}[1]
\REQUIRE Prior $q_0(\theta)$, cover-update threshold $\eta$,
         tolerance $\varepsilon$, diagnostic action set
         $\mathcal A_{\rm probe}^H$
\STATE Initialize $\thetaest_0\leftarrow \E_{q_0}[\theta]$
\STATE Compute FRR cover $\mathcal F_{\thetaest_0,\varepsilon}$ using
       \eqref{eq:frr_radius_theta}
\FOR{$t=0,1,2,\ldots$}
  \STATE Compute $\thetaest_t$ and the current posterior $q_t$
  \STATE Compute or upper-bound
         $\VoNI_\varepsilon(b_t,\thetaest_t;q_t)$ using
         \eqref{eq:voni_bound}
  \FOR{each $a^H\in\mathcal A_{\rm probe}^H$}
    \STATE Compute or approximate
           $\mathcal I_\varepsilon^H(b_t,q_t,a^H)$ using
           \eqref{eq:myopic_voni_reduction}
  \ENDFOR
  \IF{$\max_{a^H}\mathcal I_\varepsilon^H(b_t,q_t,a^H)>0$}
    \STATE Select the maximizing diagnostic action
  \ELSE
    \STATE Set $a_t^H\leftarrow\mathrm{exploit}$ and
           $u_t\leftarrow\bar\pi_{\thetaest_t,\varepsilon}(b_t)$
  \ENDIF
  \STATE Execute the selected action and observe the next diagnostic data
  \STATE Update the task belief $b_{t+1}$ and posterior $q_{t+1}(\theta)$
  \IF{$|\LV(\thetaest_t)-\LV(\thetaest_{t-1})|>\eta$}
    \STATE Recompute $\mathcal F_{\thetaest_t,\varepsilon}$
  \ENDIF
\ENDFOR
\end{algorithmic}
\end{algorithm}

\section{Convergence Analysis}
\label{sec:convergence}

\begin{theorem}[Stationary physical-noise convergence]
\label{thm:convergence}
Under Assumption~\ref{ass:slow}, suppose that $\theta_t=\thetatrue$ is constant
after some finite time, that the VoNI-guided policy visits persistently
exciting segments infinitely often, and that the FRR certificates
\eqref{eq:LQ_theta} hold uniformly in a neighborhood of $\thetatrue$.
Assume also that the posterior model is identifiable, so that distinct values
of $\theta$ induce distinct innovation/execution-residual laws on persistently
exciting segments. Then: \textup{(i)} $\thetaest_t\to\thetatrue$ almost surely;
and \textup{(ii)} for fixed FRR tolerance $\varepsilon$,
\begin{equation}
\limsup_{t\to\infty}
\left\|
\Vstar_{\thetatrue}-V^{\bar{\pi}_{\thetaest_t,\varepsilon}}_{\thetatrue}
\right\|_\infty
\le
\frac{2\varepsilon}{1-\gamma}
\quad \mathrm{a.s.}
\label{eq:limsup_policy_bound}
\end{equation}
\end{theorem}

\begin{proof}[Proof sketch]
\textbf{Part (i).}
The posterior update accumulates sufficient statistics from observation
innovations and any available execution residuals. Under persistent excitation
and identifiability, the expected log-likelihood ratio between the true
parameter $\thetatrue$ and any $\theta\neq\thetatrue$ has strictly negative
drift. Equivalently, the KL divergence between the residual law induced by
$\thetatrue$ and that induced by $\theta$ is positive on informative segments.
The likelihood ratio for every incorrect $\theta$ therefore vanishes almost
surely, and standard Bayesian consistency arguments imply posterior
concentration at $\thetatrue$. This conclusion is conditional on the
well-specified and identifiable posterior interface in
Assumption~\ref{ass:posterior_interface}; the finite-horizon estimator used in
the numerical example is evaluated empirically rather than used to establish
this asymptotic claim.

\textbf{Part (ii).}
From Part~(i), $\thetaest_t\to\thetatrue$ almost surely. Continuity of the
certified sensitivity map $\theta\mapsto\LV(\theta)$ implies
$\LV(\thetaest_t)\to\LV(\thetatrue)$ almost surely. Hence the FRR covers
calibrated at $\thetaest_t$ converge, in certified radius, to the cover
calibrated at $\thetatrue$. The mismatch terms in
Theorem~\ref{thm:voni_bound} vanish in the limit:
$\Delta_Q(\thetatrue,\thetaest_t)\to0$ and
$\rho(\thetatrue,\thetaest_t)\to1$. Therefore the excess certificate gap
vanishes, and only the ordinary FRR approximation floor remains. Applying the
FRR policy-sufficiency theorem gives~\eqref{eq:limsup_policy_bound}.

The two-timescale separation is justified because the posterior sufficient
statistics grow over time, so the outer-level posterior changes slowly after
sufficient data accumulate, while the inner-level belief $b_t$ updates at every
step. This is the standard quasi-stationary structure used in two-timescale
stochastic approximation \cite{borkar2008}.
\end{proof}

\begin{remark}[Drifting noise floors]
For gradual drift, exact convergence to a fixed $\thetatrue$ is not the
correct claim. Under bounded drift and persistent excitation, the posterior
tracks the moving physical noise floor with a nonzero steady-state error
determined by drift rate, residual variance, and the effective information
accumulated by the policy. The numerical drift experiment in
Section~\ref{sec:experiments} should therefore be interpreted as a tracking
result rather than an almost-sure convergence result.
\end{remark}

\begin{remark}[Exploration necessity]
The persistent-excitation condition in Theorem~\ref{thm:convergence} may be
supplied by ordinary task motion, by diagnostic probes, or by both. VoNI
promotes diagnostic excitation when posterior mass reaches a regime in which
the FRR certificate is at risk. As the posterior concentrates and the excess
certificate gap decreases, explicit probing can become rare. If the exploit
policy alone is not persistently exciting, however, a nonvanishing scheduled
or event-triggered probing mechanism is still required for the consistency
claim; the theorem does not justify complete cessation of informative actions
in that case.
\end{remark}

\section{Numerical Example}
\label{sec:experiments}

\subsection{Experimental Setup: UGV Dynamics, EKF, and Noise Estimation}

We evaluate VoNI in a closed-loop unicycle-model unmanned ground vehicle
(UGV) simulation with noisy position observations and an extended Kalman
filter (EKF). The purpose is not to reproduce a particular perception stack,
stereo geometry, or obstacle field. Instead, the experiment generates
innovation statistics from a closed-loop vehicle, estimator, and controller,
allowing certificate-aware noise-floor adaptation to be studied in a
dynamical setting.

The vehicle state is
\[
  x_t =
  \begin{bmatrix}
  p_{x,t} & p_{y,t} & \psi_t
  \end{bmatrix}^{\top},
\]
where $(p_{x,t},p_{y,t})$ is the planar position and $\psi_t$ is the
heading. The commanded control is
\[
  u_t =
  \begin{bmatrix}
  v_t & \omega_t
  \end{bmatrix}^{\top},
\]
where $v_t$ is forward speed and $\omega_t$ is yaw rate. The physically
executed control is
\begin{equation}
  \tilde u_t = u_t+\eta_t,
  \qquad
  \eta_t\sim\mathcal N(0,\Sigma_u),
  \label{eq:ugv_exec_noise}
\end{equation}
and the vehicle evolves according to
\begin{align}
p_{x,t+1}
&=
p_{x,t}+\Delta t\,\tilde v_t\cos\psi_t, \\
p_{y,t+1}
&=
p_{y,t}+\Delta t\,\tilde v_t\sin\psi_t, \\
\psi_{t+1}
&=
\psi_t+\Delta t\,\tilde\omega_t .
\end{align}
The observation model is a noisy position measurement
\begin{equation}
  y_t =
  \begin{bmatrix}
  p_{x,t}\\p_{y,t}
  \end{bmatrix}
  + v_t^{y},
  \qquad
  v_t^{y}\sim\mathcal N(0,\sigma_{y,t}^2I_2).
  \label{eq:ugv_obs}
\end{equation}

The EKF propagates the vehicle belief using the commanded control and
performs its measurement update using the current sensing-noise estimate
$\hat\sigma_{y,t}$. The innovation is
\[
  \nu_t = y_t-h(\hat x_t^-).
\]
To reduce contamination from state-estimation uncertainty, the scalar
innovation statistic subtracts a calibrated fraction of the EKF-predicted
position variance before updating the sensing-noise estimate. The numerical
implementation represents uncertainty in $\sigma_y$ by a truncated Gaussian
approximation
\[
  q_t(\sigma_y)
  \approx
  \mathcal N_{[0,\sigma_{\max}]}
  \bigl(\hat\sigma_{y,t},s_t^2\bigr),
\]
where $s_t$ is an adaptive posterior-spread proxy driven by innovation
surprise. Thus, the observations used for noise-floor adaptation are generated
by a closed-loop vehicle and EKF rather than by an externally prescribed
residual sequence.

The low-level task controller drives the vehicle toward
\[
  p_{\mathrm{goal}}=(10,10).
\]
When the high-level policy selects a diagnostic probe, the nominal
goal-directed command is modified by reducing its forward speed and adding a
periodic yaw excitation. The maneuver perturbs the task trajectory while
producing more informative position innovations. The posterior-entropy
baseline instead selects probing according to uncertainty in $\sigma_y$
without accounting for whether that uncertainty threatens the FRR
certificate.

The FRR action-value sensitivity is represented by a smooth calibrated proxy
$L_V(\sigma_y)$ satisfying approximately
\[
  L_V(0.05\,\mathrm{m})=1.03,
  \qquad
  L_V(0.80\,\mathrm{m})=7.97,
\]
with crossover scale
\[
  \sigma_0=0.3\,\mathrm{m}.
\]
For a candidate sensing-noise level $\sigma_y$ and current estimate
$\hat\sigma_y$, the numerical excess-certificate proxy is
\begin{equation}
\begin{aligned}
&\widehat G_\epsilon(\sigma_y,\hat\sigma_y) \\
&=
\frac{2}{1-\gamma}
\left[
 c_Q|\sigma_y-\hat\sigma_y|
 \bigl(1+L_V(\sigma_y)\bigr)
 +
 \epsilon
 \left(
 \frac{L_V(\sigma_y)}{L_V(\hat\sigma_y)}-1
 \right)
\right]_+ ,
\end{aligned}
\label{eq:numerical_certificate_proxy}
\end{equation}
where $\gamma=0.95$, $\epsilon=0.5$, and $c_Q=0.30$. This expression mirrors
the action-value mismatch and FRR-radius inflation terms in
Theorem~\ref{thm:voni_bound}. It is a calibrated certificate proxy rather
than the result of exact value iteration on the continuous-state vehicle
POMDP.

The VoNI score is approximated by posterior sampling:
\[
  \widehat{\operatorname{VoNI}}_{\epsilon,t}
  =
  \frac{1}{N_q}
  \sum_{j=1}^{N_q}
  \widehat G_\epsilon
  \bigl(\sigma_y^{(j)},\hat\sigma_{y,t}\bigr),
  \qquad
  \sigma_y^{(j)}\sim q_t,
\]
with $N_q=192$ samples. The entropy baseline uses the Gaussian
differential-entropy proxy
\[
  H_t
  =
  \frac{1}{2}\log\bigl(2\pi e\,s_t^2\bigr).
\]
Both methods convert their respective scores into stochastic probing
probabilities through smooth sigmoid rules.

We compare two complete high-level probing policies:
\begin{enumerate}
\item \textbf{VoNI-guided probing:} probing probability increases when
posterior uncertainty in $\sigma_y$ induces a large expected excess FRR
certificate gap.
\item \textbf{Posterior-entropy probing:} probing probability increases with
the entropy-like uncertainty score, irrespective of whether the uncertainty
changes the FRR certificate.
\end{enumerate}
The entropy policy therefore represents uncertainty-driven probing, whereas
VoNI represents certificate-driven probing. The comparison is between
separately calibrated complete probing policies: each includes its own smooth
trigger, probing intensity, measurement replication, and adaptation gains.
Accordingly, the results assess policy-level detection, tracking, and probing
performance rather than an ablation that changes only the trigger score.

All statistics are computed over $N=50$ Monte Carlo trials using a fixed
master seed. Reported uncertainties are standard deviations across trials.
Detection is declared when the estimated sensing noise first reaches the
crossover level $\sigma_0$ after the true step change. Drift-tracking error is
computed from step $100$ through the end of the $400$-step experiment. The
simulation code and generated data are available in the accompanying GitHub
repository~\cite{voni_code2026}.

\begin{table}[t]
\centering
\caption{UGV/EKF simulation results across $N=50$ Monte Carlo trials.}
\label{tab:ugv_summary}
\renewcommand{\arraystretch}{1.1}
\begin{tabular}{@{}lcc@{}}
\toprule
& \textbf{VoNI} & \textbf{Entropy} \\
\midrule
Step: detection time $t_{\rm det}$
& $126.48\pm2.58$ & $185.08\pm7.32$ \\
Step: post-jump delay
& $26.48\pm2.58$ & $85.08\pm7.32$ \\
Step: probing fraction
& $13.93\pm1.16\%$ & $95.21\pm1.09\%$ \\
Drift: tracking MAE (m)
& $0.1407\pm0.0150$ & $0.2394\pm0.0172$ \\
Drift: probing fraction
& $39.96\pm1.87\%$ & $96.77\pm0.75\%$ \\
\bottomrule
\end{tabular}
\end{table}

\subsection{Certificate-Gap Landscape}

Figure~\ref{fig:ugv_certificate_landscape} shows the
certificate-gap proxy as a function of the sensing-noise estimate
$\hat\sigma_y$ for four candidate true sensing-noise levels. The gap is small
when the candidate and estimated noise floors are both in the nominal
sub-crossover regime. In contrast, when the candidate noise floor is degraded
but the estimate remains nominal, the gap is large because the FRR cover is
calibrated to an under-resolved operating condition.

The figure visualizes the principal mechanism behind VoNI: uncertainty in the
physical noise floor is decision-relevant only when it changes the calibrated
action-value sensitivity or invalidates the current reliability-cover radius.

\begin{figure}[t]
\centering
\includegraphics[width=\columnwidth]
{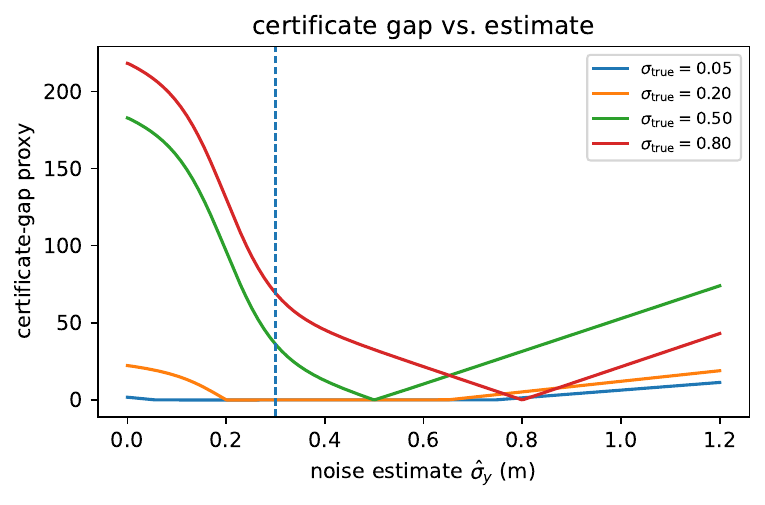}
\caption{
Certificate-gap proxy as a function of the sensing-noise estimate
$\hat\sigma_y$ for several candidate true sensing-noise floors. Posterior
uncertainty is valuable when it places substantial probability on regimes that
can invalidate the current FRR certificate.
}
\label{fig:ugv_certificate_landscape}
\end{figure}

\subsection{Scenario 1: Step Sensing-Noise Jump}

The step-jump scenario runs for $300$ steps. For steps $0$--$99$, the sensing
noise floor is nominal:
\[
  \sigma_{y,\mathrm{true}}=0.05\,\mathrm{m}.
\]
At step $100$, the true noise floor changes abruptly to
\[
  \sigma_{y,\mathrm{true}}=0.80\,\mathrm{m}.
\]
No external regime indicator is supplied; the change must be inferred from
the EKF innovation sequence.

Figure~\ref{fig:ugv_step_results}(a) shows the estimated sensing-noise floor.
VoNI detects the jump at
\[
  t_{\rm det}=126.48\pm2.58,
\]
whereas the posterior-entropy policy detects it at
\[
  t_{\rm det}=185.08\pm7.32.
\]
Since the jump occurs at step $100$, the corresponding post-jump delays are
\[
  26.48\pm2.58
  \qquad\text{and}\qquad
  85.08\pm7.32
\]
steps, respectively. VoNI therefore reduces the mean post-jump detection
delay by approximately $69\%$.

Despite detecting the change earlier, VoNI probes during only
$13.93\pm1.16\%$ of the experiment, whereas the entropy policy probes during
$95.21\pm1.09\%$ of the steps. Thus, the VoNI policy reduces probing frequency
by approximately $85\%$ relative to the entropy policy.

Figure~\ref{fig:ugv_step_results}(b) shows the corresponding FRR
certificate-gap proxy. Before the jump, both policies have near-zero gap
because the nominal sensing model is appropriate. After the jump, the
VoNI score increases when the degraded regime threatens the reliability
certificate. The resulting diagnostic probing accelerates adaptation and
reduces the certificate gap earlier. The entropy policy probes much more
frequently but does not focus its probing specifically on
certificate-relevant mismatch.

\begin{figure}[t]
\centering
\includegraphics[width=\columnwidth]
{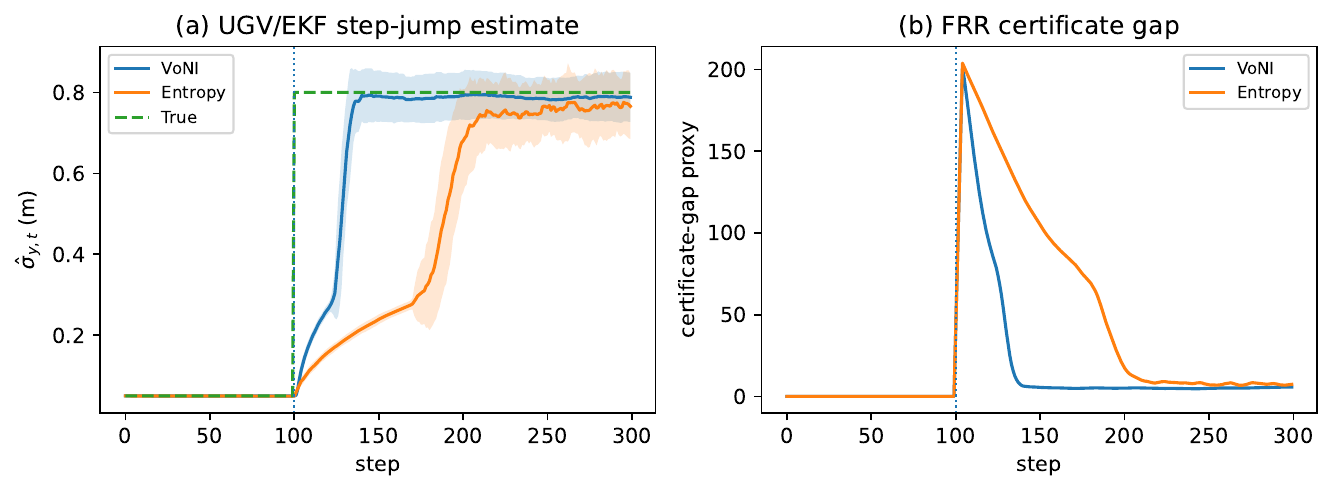}
\caption{
UGV/EKF step-jump experiment.
\textbf{(a)} Estimated sensing-noise floor $\hat\sigma_{y,t}$.
The true noise changes from $0.05\,\mathrm{m}$ to $0.80\,\mathrm{m}$ at
step $100$.
\textbf{(b)} FRR certificate-gap proxy. VoNI detects the degraded sensing
regime earlier and reduces the mismatch-induced certificate gap faster than
posterior-entropy probing.
}
\label{fig:ugv_step_results}
\end{figure}

\subsection{Scenario 2: Gradual Sensing-Noise Drift}

The gradual-drift scenario runs for $400$ steps. The true sensing-noise floor
is nominal for the first $100$ steps, increases linearly from
$0.05\,\mathrm{m}$ to $0.80\,\mathrm{m}$ between steps $100$ and $350$, and
then remains at the degraded level.

Figure~\ref{fig:ugv_drift_results}(a) shows that both methods respond to the
increasing sensing-noise floor, but VoNI achieves lower tracking error:
\[
\begin{aligned}
  \mathrm{MAE}_{\mathrm{VoNI}}
  &=
  0.1407\pm0.0150\,\mathrm{m},\\
  \mathrm{MAE}_{\mathrm{Entropy}}
  &=
  0.2394\pm0.0172\,\mathrm{m}.
\end{aligned}
\]
Thus, VoNI reduces mean drift-tracking error by approximately $41\%$.

Figure~\ref{fig:ugv_drift_results}(b) shows the rolling diagnostic-probe
rate. VoNI probes during
\[
  39.96\pm1.87\%
\]
of the drift-scenario steps, whereas the posterior-entropy policy probes
during
\[
  96.77\pm0.75\%.
\]
VoNI therefore uses approximately $59\%$ fewer probing actions while also
achieving lower tracking error.

\begin{figure}[t]
\centering
\includegraphics[width=\columnwidth]
{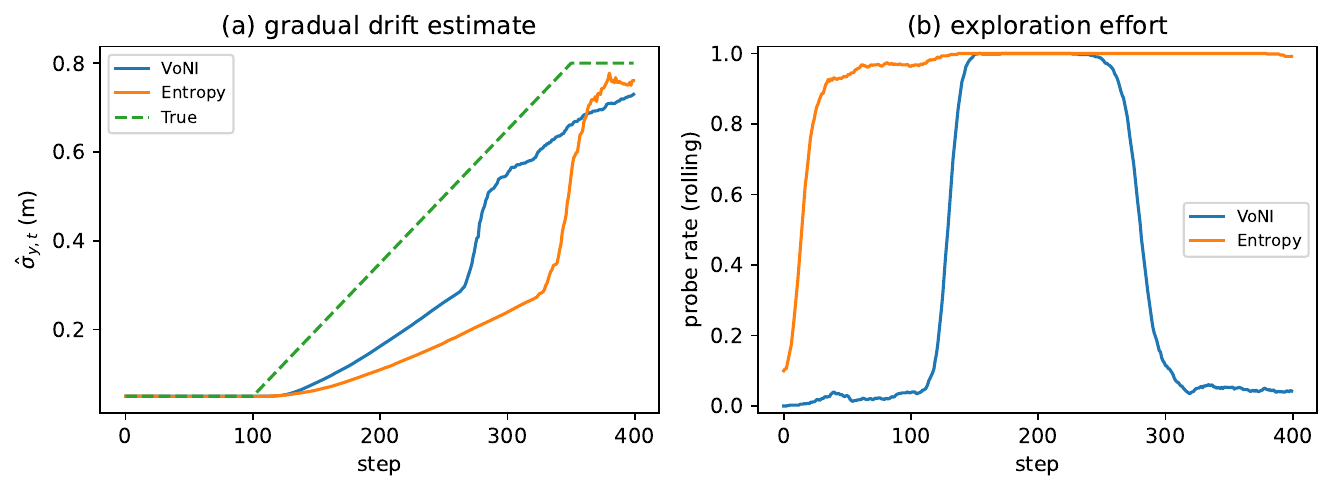}
\caption{
UGV/EKF gradual-drift experiment.
\textbf{(a)} Estimated sensing-noise floor during gradual degradation.
\textbf{(b)} Rolling diagnostic-probe rate. VoNI achieves lower tracking
error while using substantially fewer probing actions.
}
\label{fig:ugv_drift_results}
\end{figure}

\subsection{Closed-Loop Trajectory and Diagnostic Probing}

Figure~\ref{fig:ugv_trajectory} shows representative closed-loop UGV
trajectories for the VoNI and posterior-entropy policies. Both drive the
vehicle from the initial condition toward
$p_{\mathrm{goal}}=(10,10)$, but their diagnostic behaviors differ. VoNI
inserts probes selectively when posterior uncertainty threatens the FRR
certificate. The entropy policy probes much more frequently because it reacts
to uncertainty without considering its effect on the reliability cover.

The trajectories illustrate the operational interpretation of VoNI: probing
is not intrinsically valuable. It is valuable only when the expected
certificate improvement justifies temporarily perturbing the task trajectory.

\begin{figure}[t]
\centering
\includegraphics[width=\columnwidth]
{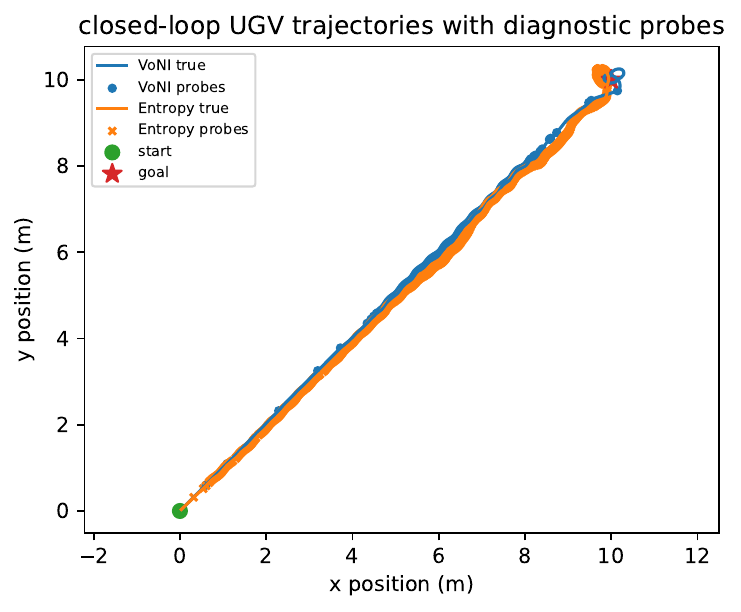}
\caption{
Representative closed-loop UGV trajectories with diagonal goal
$p_{\mathrm{goal}}=(10,10)$ and diagnostic probes. VoNI uses probing actions
selectively, whereas the posterior-entropy policy probes more persistently.
The plot visualizes closed-loop diagnostic behavior rather than a specific
competition-map geometry.
}
\label{fig:ugv_trajectory}
\end{figure}

\subsection{Discussion}

The UGV/EKF experiments support three claims. First, VoNI detects an abrupt
sensing-noise jump substantially earlier than posterior-entropy probing, reducing
post-jump detection delay by approximately $69\%$. Second, in the gradual
drift scenario, VoNI achieves lower tracking error while using approximately
$59\%$ fewer probing actions than the posterior-entropy baseline. Third, the trajectory-level
simulation shows that these gains arise from selective diagnostic probing in a
closed-loop vehicle/filter system, not from an externally supplied residual
sequence.

These results should still be interpreted as certificate-level validation
rather than full-stack autonomous-vehicle validation. The experiment uses a
simple unicycle model, position observations, and an EKF, and the FRR
suboptimality is represented by a calibrated certificate-gap proxy rather than
by exact POMDP value iteration. Nevertheless, the simulation captures the
mechanism needed for the present paper: the innovation sequence is generated by
a physical vehicle model, and exploration is triggered only when uncertainty in
the physical noise floor threatens the FRR certificate.

\section{Conclusion}
\label{sec:conclusion}

We introduced the Value of Noise Information (VoNI) as a principled,
FRR-grounded criterion for active physical noise-floor estimation in partially
observed autonomous systems. The latent parameter may be a scalar sensing noise
floor or a combined sensing--execution parameter
$\theta=(\sigma_y,\sigma_u)$ that includes actuator or action-execution
uncertainty. VoNI quantifies the excess FRR certificate gap caused by using a
reliability cover calibrated to an incorrect physical noise parameter, is
bounded analytically by the Lipschitz machinery of \cite{frr2026}, and is
small when the posterior credible region remains in a sub-crossover regime
where the certificate is insensitive to $\theta$.
The bi-level decision maker derived from VoNI achieves almost-sure convergence
of the physical noise estimate in stationary regimes and convergence of the
policy loss to the FRR approximation floor.
The high-level probing rule compares the expected reduction in the excess
FRR certificate gap with the task-level cost of diagnostic action; the
numerical implementation uses a smooth posterior-sampling approximation to
this rule. The same construction admits a high-level finite-POMDP
interpretation: VoNI is a certificate-aware myopic
approximation to the value of disambiguating the latent sensing--execution
regime.
Numerical results instantiate the sensor-only case on a closed-loop UGV/EKF
model and confirm substantially earlier detection of a sensing-noise jump than
posterior-entropy baselines, with improved drift tracking and fewer unnecessary
probing actions.

Future work will extend the framework to finite-sample convergence rates,
which would bound detection delay in terms of problem parameters such as the
crossover scale, $\gamma$, and the Fisher information. A second direction is
joint sensing--execution estimation, where observation innovations estimate
$\sigma_y$ and motion or tracking residuals estimate $\sigma_u$. Multi-modal
posteriors arising from categorical terrain types and integration with the
NLDCA framework for coarse-to-fine planning under context-dependent noise are
natural next steps.

\bibliographystyle{IEEEtran}
\bibliography{references}

@inproceedings{haarnoja2018sac,
  title     = {Soft Actor-Critic: Off-Policy Maximum Entropy Deep Reinforcement Learning with a Stochastic Actor},
  author    = {Haarnoja, Tuomas and Zhou, Aurick and Abbeel, Pieter and Levine, Sergey},
  booktitle = {Proceedings of the 35th International Conference on Machine Learning},
  pages     = {1861--1870},
  year      = {2018},
  volume    = {80},
  series    = {Proceedings of Machine Learning Research},
  publisher = {PMLR}
}

@misc{voni_code2026,
  author       = {Yoon, Hyung-Jin},
  title        = {{VoNI POMDP Noise Estimation}: Simulation Code},
  year         = {2026},
  howpublished = {\url{https://github.com/LCAS-Lab/voni-pomdp-noise-estimation}},
  note         = {GitHub repository}
}

@misc{frr2026,
  title         = {Finite Reliability Representations: Noise-Calibrated Belief-Space Covers for Reliable Decision-Making},
  author        = {Hyung-Jin Yoon and Hunmin Kim},
  year          = {2026},
  eprint        = {2607.04019},
  archivePrefix = {arXiv},
  primaryClass  = {eess.SY},
  url           = {https://arxiv.org/abs/2607.04019},
  doi           = {10.48550/arXiv.2607.04019}
}

@article{kaelbling1998,
  author  = {Kaelbling, Leslie Pack and Littman, Michael L. and Cassandra, Anthony R.},
  title   = {Planning and Acting in Partially Observable Stochastic Domains},
  journal = {Artificial Intelligence},
  volume  = {101},
  number  = {1--2},
  pages   = {99--134},
  year    = {1998}
}

@incollection{spaan2012,
  author    = {Spaan, Matthijs T. J.},
  title     = {Partially Observable Markov Decision Processes},
  booktitle = {Reinforcement Learning: State-of-the-Art},
  editor    = {Wiering, Marco and van Otterlo, Martijn},
  publisher = {Springer},
  address   = {Berlin, Heidelberg},
  pages     = {387--414},
  year      = {2012}
}

@article{hauskrecht2000,
  author  = {Hauskrecht, Milos},
  title   = {Value-Function Approximations for Partially Observable Markov Decision Processes},
  journal = {Journal of Artificial Intelligence Research},
  volume  = {13},
  pages   = {33--94},
  year    = {2000}
}

@inproceedings{pineau2003,
  author    = {Pineau, Joelle and Gordon, Geoff and Thrun, Sebastian},
  title     = {Point-Based Value Iteration: An Anytime Algorithm for {POMDP}s},
  booktitle = {Proceedings of the International Joint Conference on Artificial Intelligence},
  pages     = {1025--1030},
  year      = {2003}
}

@inproceedings{kurniawati2008,
  author    = {Kurniawati, Hanna and Hsu, David and Lee, Wee Sun},
  title     = {{SARSOP}: Efficient Point-Based {POMDP} Planning by Approximating Optimally Reachable Belief Spaces},
  booktitle = {Proceedings of Robotics: Science and Systems},
  year      = {2008}
}

@inproceedings{roy2002,
  author    = {Roy, Nicholas and Gordon, Geoffrey J.},
  title     = {Exponential Family {PCA} for Belief Compression in {POMDP}s},
  booktitle = {Advances in Neural Information Processing Systems},
  volume    = {15},
  pages     = {1667--1674},
  year      = {2002}
}

@article{mehra1970,
  author  = {Mehra, Raman K.},
  title   = {On the Identification of Variances and Adaptive Kalman Filtering},
  journal = {IEEE Transactions on Automatic Control},
  volume  = {15},
  number  = {2},
  pages   = {175--184},
  year    = {1970}
}

@inproceedings{ross2007bayesadaptive,
  author    = {Ross, St{\'{e}}phane and Chaib-draa, Brahim and Pineau, Joelle},
  title     = {Bayes-Adaptive {POMDP}s},
  booktitle = {Advances in Neural Information Processing Systems},
  year      = {2007}
}

@article{doshi2010,
  author  = {Doshi-Velez, Finale and Pfau, David and Wood, Frank and Roy, Nicholas},
  title   = {Bayesian Nonparametric Methods for Partially-Observable Reinforcement Learning},
  journal = {IEEE Transactions on Pattern Analysis and Machine Intelligence},
  volume  = {37},
  number  = {2},
  pages   = {394--407},
  year    = {2015}
}

@article{lindley1956,
  author  = {Lindley, Dennis V.},
  title   = {On a Measure of the Information Provided by an Experiment},
  journal = {The Annals of Mathematical Statistics},
  volume  = {27},
  number  = {4},
  pages   = {986--1005},
  year    = {1956}
}

@article{mackay1992,
  author  = {MacKay, David J. C.},
  title   = {Information-Based Objective Functions for Active Data Selection},
  journal = {Neural Computation},
  volume  = {4},
  number  = {4},
  pages   = {590--604},
  year    = {1992}
}

@article{krause2008,
  author  = {Krause, Andreas and Singh, Ajit and Guestrin, Carlos},
  title   = {Near-Optimal Sensor Placements in {Gaussian} Processes: Theory, Efficient Algorithms and Empirical Studies},
  journal = {Journal of Machine Learning Research},
  volume  = {9},
  pages   = {235--284},
  year    = {2008}
}

@inproceedings{bourgault2002,
  author    = {Bourgault, Fr{\'{e}}d{\'{e}}ric and Makarenko, Alexei A. and Williams, Stefan B. and Grocholsky, Ben and Durrant-Whyte, Hugh F.},
  title     = {Information Based Adaptive Robotic Exploration},
  booktitle = {Proceedings of the IEEE/RSJ International Conference on Intelligent Robots and Systems},
  pages     = {540--545},
  year      = {2002}
}

@inproceedings{charrow2015,
  author    = {Charrow, Benjamin and Liu, Sikang and Kumar, Vijay and Michael, Nathan},
  title     = {Information-Theoretic Mapping Using Cauchy-Schwarz Quadratic Mutual Information},
  booktitle = {Proceedings of the IEEE International Conference on Robotics and Automation},
  pages     = {4791--4798},
  year      = {2015}
}

@article{sutton1999options,
  author  = {Sutton, Richard S. and Precup, Doina and Singh, Satinder},
  title   = {Between {MDP}s and Semi-{MDP}s: A Framework for Temporal Abstraction in Reinforcement Learning},
  journal = {Artificial Intelligence},
  volume  = {112},
  number  = {1--2},
  pages   = {181--211},
  year    = {1999}
}

@book{borkar2008,
  author    = {Borkar, Vivek S.},
  title     = {Stochastic Approximation: A Dynamical Systems Viewpoint},
  publisher = {Cambridge University Press},
  year      = {2008}
}

@misc{adaptive_mppi_noise2026,
  author        = {Hyung-Jin Yoon and Hunmin Kim},
  title         = {Adaptive {MPPI} with Online Disturbance Covariance Estimation:
                   Provable Stability Tightening via Spatial Smoothing},
  year          = {2026},
  eprint        = {2607.08942},
  archivePrefix = {arXiv},
  primaryClass  = {eess.SY},
  url           = {https://arxiv.org/abs/2607.08942},
  doi           = {10.48550/arXiv.2607.08942}
}

\end{document}